\documentclass[11pt]{article}
\usepackage{amsthm,amsfonts,amsmath,times}
\usepackage{enumerate}
\usepackage{algo, extras, xspace, color}
\usepackage{algorithm,algorithmic}
\usepackage{boxedminipage, wrapfig}
\usepackage{hyperref,url}
\usepackage{tikz}
\usepackage{graphicx}
\usetikzlibrary{arrows}

\advance \topmargin by -1.0in
\advance \oddsidemargin by -0.8in
\newcommand{\myparskip}{3pt}
\parskip \myparskip

\newtheorem{lemma}{Lemma}[section]
\newtheorem{theorem}[lemma]{Theorem}

\newtheorem{corollary}[lemma]{Corollary}
\newtheorem{prop}[lemma]{Proposition}

\newtheorem{remark}[lemma]{Remark}
\newtheorem{prob}{Problem}

\newtheorem{question}{Question}

\newcommand{\etal}{{\em et al.}~}
\DeclareMathOperator*{\Ex}{\mathbb{E}}

\renewenvironment{proof}{\vspace{-0.1in}\noindent{\bf Proof:}}%
        {\hspace*{\fill}$\Box$\par}
        {\hspace*{\fill}$\Box$\par}
        {\hspace*{\fill}$\Box$\par}

\def\eps{\varepsilon}

\def\bx{\textbf{x}}

\def\MC{\textsc{Edge-wt-MC}\xspace}
\def\DirMC{\textsc{Dir-MC}\xspace}
\def\2DirMC{\text{$st$-Bi-Cut}\xspace}
\def\NodeMC{\textsc{Node-wt-MC}\xspace}
\def\DirMCRel{\textsc{Dir-MC-Rel}\xspace}

\def\NodeMC{\textsc{Node-wt-MC}\xspace}
\def\NodeMCRel{\textsc{Node-MC-Rel}\xspace}
\def\DirMulticut{\textsc{Dir-Multicut}\xspace}

\def\SubMP{\textsc{Sub-MP}\xspace}

\def\SymSubMP{\textsc{Sym-Sub-MP}\xspace}

\def\AHMCfull{\textsc{Hypergraph Multiway Cut}\xspace}
\def\AHMC{\textsc{Hypergraph-MC}\xspace}

\newcommand{\cP}{\mathcal{P}}

\begin{document}
\title{Simple and Fast Rounding Algorithms for Directed and Node-weighted Multiway Cut}

\author{
Chandra Chekuri
\thanks{
Dept. of Computer Science, University of Illinois, Urbana, IL 61801.
Supported in part by NSF grant CCF-1319376. {\tt chekuri@illinois.edu}}
\and
Vivek Madan
\thanks{
Dept. of Computer Science, University of Illinois, Urbana, IL 61801.
Supported in part by NSF grant CCF-1319376.
{\tt vmadan2@illinois.edu}}
}

\date{\today}

\maketitle

\thispagestyle{empty}
\begin{abstract}
  We study the multiway cut problem in {\em directed} graphs and one
  of its special cases, the {\em node-weighted} multiway cut problem
  in {\em undirected} graphs.  In {\sc Directed Multiway Cut} (\DirMC)
  the input is an edge-weighted directed graph $G=(V,E)$ and a set of
  $k$ terminal nodes $\{s_1,s_2,\ldots,s_k\} \subseteq V$; the goal is
  to find a min-weight subset of edges whose removal ensures that there
  is no path from $s_i$ to $s_j$ for any $i \neq j$. In {\sc
    Node-weighted Multiway Cut} (\NodeMC) the input is a node-weighted
  undirected graph $G$ and a set of $k$ terminal nodes
  $\{s_1,s_2,\ldots,s_k\} \subseteq V$; the goal is to remove a
  min-weight subset of nodes to disconnect each pair of
  terminals. \DirMC admits a $2$-approximation \cite{NaorZ01}
  and \NodeMC admits a $2(1-\frac{1}{k})$-approximation
  \cite{GargVY04}, both via rounding of LP relaxations.
  Previous rounding algorithms for these problems, from nearly twenty
  years ago, are based on careful rounding of an {\em
    optimum} solution to an LP relaxation. This is particularly true
  for \DirMC for which the rounding relies on a custom LP formulation
  instead of the natural distance based LP relaxation \cite{NaorZ01}.

  In this paper we describe extremely simple and near linear-time
  rounding algorithms for \DirMC and \NodeMC via a natural distance
  based LP relaxation. The dual of this relaxation is a special case
  of the maximum  multicommodity flow problem. Our
  algorithms achieve the same bounds as before but have the
  significant advantage in that they can work with {\em any feasible}
  solution to the relaxation. Consequently, in addition to obtaining
  ``book'' proofs of LP rounding for these two basic problems, we also
  obtain significantly faster approximation algorithms by taking
  advantage of known algorithms for computing near-optimal solutions
  for maximum multicommodity flow problems. We also investigate 
  lower bounds for \DirMC when $k=2$ and in particular prove that
  the integrality gap of the LP relaxation is $2$ even in directed
  planar graphs.
\end{abstract}

\newpage
\section{Introduction}
We study several variants of the multiway cut problem in graphs (also
referred to as the mult-terminal cut problem). In the classical
$s$-$t$ cut problem the input consists of a graph $G=(V,E)$ and two
distinct nodes $s,t$; the goal is to separate $s$ from $t$ by removing
a minimum cost set of edges and/or nodes. In the multiway cut problem
the input is a graph $G=(V,E)$ and a set $S = \{s_1,s_2,\ldots,s_k\}$
of $k$ nodes from $V$ called terminals; the goal is to separate the terminals
from each other at minimum cost by removing edges and/or nodes. We
describe the three main variants that are of interest to us.

\medskip 
\noindent {\sc Multiway Cut} (\MC): The input is an undirected
graph $G=(V,E)$ along with non-negative edge weights $w(e), e \in E$
and a set $\{s_1,\ldots,s_k\} \subseteq V$ of terminals. The goal
is to find a min-cost set of edges $E' \subseteq E$ such that in $G
- E'$ there is no path from $s_i$ to $s_j$ for $i \neq j$.

\medskip 
\noindent {\sc Node-Weighted Multiway Cut} (\NodeMC): The input is an
undirected graph $G=(V,E)$ along with non-negative node weights $w(v),
v \in V$ and a set $\{s_1,\ldots,s_k\} \subseteq V$ of terminals. The
goal is to find a min-cost set of nodes $V' \subseteq V$ such that
in $G - V'$ there is no path from $s_i$ to $s_j$ for $i \neq j$.\footnote{
In this definition terminals are allowed to be removed. If they are not allowed
to be removed we can simply make their weight $\infty$.}

\medskip 
\noindent {\sc Directed Multiway Cut} (\DirMC): The input is a
directed graph $G=(V,E)$ along with non-negative edge weights $w(e),
e \in E$ and a set $\{s_1,\ldots,s_k\} \subseteq V$ of terminals. The
goal is to find a min-cost set of edges $E' \subseteq E$ such that
in $G - E'$ there is no path from $s_i$ to $s_j$ for $i \neq j$.

\begin{remark}
  \DirMC with $k=2$ is \emph{not} the same as the $s$-$t$ cut problem.
  The goal is to separate $s_1$ from $s_2$ \emph{and} $s_2$ from $s_1$.
  In fact \DirMC with $k=2$ is NP-Hard~\cite{GargVY94}. 
\end{remark}

The complexity of the multiway cut problem and its variants have been
extensively studied since the paper of Dahlhaus \etal
\cite{DahlhausJPSY92}. They showed that \MC with $k=3$ is NP-Hard; it
was later observed that the problem is also APX-hard to approximate. This
is in contrast to the case of $k=2$ which can be solved in 
polynomial-time in undirected graphs via a reduction to the
$s$-$t$ minimum-cut problem.

\MC reduces in an approximation preserving fashion to \NodeMC which in
turn reduces in an approximation preserving fashion to \DirMC
\cite{GargVY04}; it is also easy to see that in the directed case,
node-weighted and edge-weighted versions are equivalent. The current
best approximation ratio for \MC stands at $1.2965$ due to Sharma and
Vondr\'ak \cite{SharmaV14}. For \NodeMC a $2(1-1/k)$ approximation is
known from the work of Garg, Vazirani and Yannakakis \cite{GargVY04},
and for \DirMC a $2$ approximation is known from the work of Naor and
Zosin \cite{NaorZ01}.  {\sc Vertex Cover} reduces to \NodeMC and
\DirMC in an approximation preserving
fashion~\cite{GargVY04}. Assuming $P \neq NP$ {\sc Vertex Cover} is
hard to approximate to within a factor of $1.36$ \cite{DinurS05}, and
assuming the Unique Games Conjecture it is hard to approximate to
within a factor of $(2-\eps)$ for any fixed $\eps > 0$
\cite{KhotR08}. These hardness results apply to \NodeMC and \DirMC and
show that \MC is provably easier to approximate than them.

Our focus in this paper is on approximation algorithms for \NodeMC and
\DirMC. The known algorithms are based on rounding suitable LP
relaxations for the problems. For both problems there is a simple and
natural LP relaxation based on distance variables on nodes/edges; see
Section~\ref{sec:dir-mc} and \ref{sec:node-mc}. (We note that a
similar relaxation applies to the more general {\sc Multicut} problem
and that dual of the LP relaxation corresponds to the LP for maximum
multicommodity flow.)  For \NodeMC the algorithm of Garg,
Vazirani and Yannakakis \cite{GargVY04} shows that any {\em optimum}
solution to the relaxation can be converted to a half-integral optimum
solution which can then be rounded easily.  The situation for \DirMC
is much more involved. Unlike the case of \NodeMC, half-integral
optimum solutions may not exist for the relaxation even for $k=2$.
Garg \etal \cite{GargVY94} obtained an $O(\log k)$-approximation via
the relaxation using ideas from approximation algorithms for multicut
\cite{GVY}.  Naor and Zosin obtained a $2$-approximation for \DirMC in
an elegant, surprising and somewhat mysterious fashion. They write a
different LP relaxation called the {\em relaxed multiway flow}
relaxation which is within a factor of $2$ of the natural relaxation,
and show that an {\em optimum} solution to this new relaxation can be
rounded without any loss in the approximation. This gives an indirect
proof that the natural relaxation has an integrality gap of at most
$2$. The proof of correctness crucially relies on complementary
slackness properties of the optimum solution and is partly inspired by
the ideas in \cite{GargVY04}. The idea of using a relaxed multiway
flow is inspired by earlier work on the subset feedback vertex problem
\cite{EvenNZ00}.

The algorithms of \cite{GargVY04} and \cite{NaorZ01} are from almost
twenty years ago. During this intervening years no alternative
algorithms or rounding schemes have been obtained for these basic
problems. We observe that for the case of \MC there is an extremely
simple rounding scheme that converts any fractional feasible solution
to a multiway cut with a loss of a factor of $2$ (see
\cite{Vazirani-book}). The algorithm picks a random $\theta \in
(0,1/2)$ and for each terminal $s_i$ removes the edges leaving the
ball $B(s_i,\theta)$ of nodes contained within a radius $\theta$
around $s_i$ (with respect to distances given by the LP solution);
more formally the output is $\bigcup_{i=1}^k \delta(B(s_i,\theta))$.

In this paper we show that very simple algorithms which are essentially
similar in spirit to the above scheme also work for \DirMC and \NodeMC!
\begin{itemize}
\item The rounding algorithms are extremely simple and natural to
  describe, and in retrospect also to analyze. 
\item The algorithms only require a feasible solution to the natural
  LP relaxation and not necessarily an optimum solution.
\item Given a feasible fractional solution, the rounding algorithms
  can be implemented in time that is similar to what is required for
  one single-source shortest path computation. The deterministic
  version requires an additional logarithmic factor.
\end{itemize}

In addition to algorithmic results we also obtain some lower bound
results for \DirMC with $k=2$; the goal is to separate $s$ from $t$
{\em and} $t$ from $s$ in a directed graph $G$; subsequently 
we refer to this special case as \2DirMC. We prove that the
natural LP relaxation has an integrality gap of $2$ for
\2DirMC even in {\em planar} directed graphs.

We believe that our algorithms and analysis will be useful for related
problems. Indeed one of our motivations for simplifying the rounding
schemes for \DirMC and \NodeMC came from attempts to obtain algorithms
for a problem with applications to network information theory
\cite{KKCV15}. A significant consequence of our rounding algorithms
are much faster approximation algorithms for \NodeMC and \DirMC in
both theory and practice. Solving the LP relaxations for \NodeMC and
\DirMC to optimality is quite challenging. The options are to use the
Ellipsoid method or to use a compact formulation with a very large
number of variables and constraints. As we remarked earlier, the dual
of the natural LP relaxation for these problems is the maximum
multicommodity flow problem. Combinatorial fully-polynomial time
approximation schemes for solving these multicommodity flow problems
have been extensively investigated in theoretical computer science and
mathematical programming with a number of techniques developed over
the years; we refer the reader to
\cite{PST95,GrigoriadisK94,Young95,Bienstock-book,GargK,Fleischer,BienstockI06,Madry10}. Thus,
a fast $(1+\eps)$-approximation for the LP relaxation for \NodeMC and
\DirMC can be obtained using these methods. The fastest theoretical
algorithms run in time $\tilde{O}(m^2/\eps^2)$ \cite{Fleischer,GargK}
or in even faster $\tilde{O}(mn/\eps^2)$ time \cite{Madry10} under
some mild conditions; here $m$ is the number of edges and $n$ is the
number of nodes in $G$ and $\tilde{O}$ suppresses poly-logarithmic
factors. Note that these running times are independent of $k$.  Our
rounding algorithms can convert such an approximate feasible solution
to an integral cut in near-linear time with a factor of $2$ loss in
the cost.  Thus, we can obtain provably fast $(2+\eps)$-approximation
algorithms. Since our focus is on the rounding algorithms we do not go
into further details of specific algorithms or running times for
solving the relaxation.

We refer the interested reader to quickly jump to
Section~\ref{sec:dir-mc} to see the simplicity of the rounding scheme
and its analysis for \DirMC that achieves a bound of $2$. This also
applies to \NodeMC via a simple reduction to \DirMC. We also discuss
some new observations on the hardness of the problem when $k=2$. In
Section~\ref{sec:node-mc} we give a slightly different rounding scheme
for \NodeMC that achieves an improved bound of $2(1-1/k)$, matching
the known ratio from \cite{GargVY04}.

\subsection{Other related work}
\label{subsec:related-work}

The natural LP relaxation for \MC has an integrality gap of $2(1-1/k)$.
Approximation algorithms for \MC received substantial attention
following the breakthrough work of Calinescu, Karloff and Rabani
\cite{CalinescuKR98}. They developed a new ``geometric'' LP relaxation
(henceforth referred to as the CKR-relaxation) which they used to
obtain a $(1.5-1/k)$-approximation. The integrality gap of the
CKR-relaxation, and consequently the approximation ratio, was improved
subsequently to $1.3438$ by Karger \etal \cite{KargerKSTY99}, to
$1.32388$ by Buchbinder \etal \cite{BuchbinderNS13}, and to 
the currently best known bound of $1.2965$ by Sharma and Vondr\'ak
\cite{SharmaV14}. For $k=3$ a tight bound of $12/11$ is known
\cite{CheungCT06,KargerKSTY99}. It is also known that assuming the
Unique Games Conjecture, for any fixed $k$, the approximability
threshold for \MC coincides with the integrality gap of the
CKR-relaxation \cite{ManokaranNRS08}.

The CKR-relaxation makes use of the observation that \MC can be viewed
as a partition problem where the goal is to partition the node set
$V(G)$ into $k$ parts $V_1,\ldots,V_k$ to minimize $\sum_{i=1}^k
w(\delta(V_i))$ subject to the constraint that for $1 \le i \le k$,
$s_i \in V_i$.  {\sc Submodular Multiway Partition} (\SubMP) is a
generalization from the setting of graphs to arbitrary submodular
functions. Here we are given a non-negative submodular function $f:2^V
\rightarrow \mathbb{R}^+$ over the ground set $V$ along with terminals
$\{s_1,\ldots,s_k\} \subset V$.  The goal is to partition $V$ into
$V_1,\ldots,V_k$ to minimize $\sum_{i=1}^k f(V_i)$ subject to the
constraint that $s_i \in V_i$ for $1 \le i \le k$.  If $f$ is
symmetric, as in the case of the undirected graph cut function, we
obtain the {\sc Symmetric Submodular Multiway Partition} (\SymSubMP)
problem. These problems were considered by Zhao, Nagamochi and Ibaraki
\cite{ZhaoNI05} who analyzed greedy-splitting algorithms, and more
recently by Chekuri and Ene \cite{ChekuriE11b} who used a
Lov\'asz-extension based convex relaxation. Interestingly, the convex
relaxation when specialized to \MC yields the CKR-relaxation.Chekuri and Ene
\cite{ChekuriE11b} obtained a $(1.5-1/k)$-approximation for \SymSubMP
and $2$-approximation for \SubMP. Ene, Vondr\'ak and Wu \cite{EneVW13}
improved the bound for \SubMP to $2(1-1/k)$ and also obtained
lower bound results in the oracle model.

\NodeMC cannot be viewed as a partition problem
directly. Nevertheless, it can be seen that \NodeMC
is equivalent to \AHMCfull problem (\AHMC) which is a generalization
of \MC from graphs to hypergraphs. \AHMC can be
cast as a special case of \SubMP (note that the reduction uses a
non-symmetric submodular function $f$) and thus \NodeMC can be
indirectly reduced to a partition problem. This leads to an
alternative $2(1-1/k)$-approximation for \NodeMC based on the
Lov\'asz-extension based relaxation for \AHMC. This relaxation does
not result in a better worst-case approximation than the
distance-based relaxation, however, it appears to be strictly stronger
in that it improves the approximation ratio in special some cases as
observed in \cite{ChekuriE11}. No fast approximation
algorithms are known to solve this convex relaxation.

Finally we mention the {\sc Multicut} problem where the goal is to
separate a given set of $k$ node-pairs $(s_1,t_1),\ldots,(s_k,t_k)$ in
a given graph at minimum-cost. One can consider undirected graphs with
edge weights, undirected graphs with node weights and directed graph
with edge weights.  These versions generalize the corresponding
multiway cut problems. The best known approximation ratio for {\sc
  Multicut} in undirected graphs is $O(\log k)$ \cite{GVY,GargVY94}
while the best known bounds in directed graphs is $\min(k,
\tilde{O}(n^{11/23}))$ \cite{AgarwalAC07}. Moreover, it is known from
the work of Chuzhoy and Khanna \cite{ChuzhoyK09} that the problem in
directed graphs is inapproximable to a factor better than
$\tilde{\Omega}(2^{\log^{1-\eps}n})$.

\section{LP Relaxation and rounding for \DirMC}
\label{sec:dir-mc}

\DirMC can be naturally formulated as an integer linear program with
variables $x_e \in \{0,1\}$, $e \in E$ which indicate whether $e$ is
cut or not. Let $\cP_{ij}$ be the set of all directed paths from $s_i$
to $s_j$ in $G$.  The constraint that $s_i$ is separated from $s_j$ by
the cut can be enforced by requiring that $\sum_{e \in p} x_e \ge 1$
for each $p \in \cP_{ij}$.  This leads to the following LP relaxation
where the integer constraint $x_e \in \{0,1\}$ is replaced by $x_e \in
[0,1]$. We can without loss of generality drop the constraint $x_e \le
1$.

\begin{figure}[htb]
  \centering
\begin{boxedminipage}{0.5\linewidth}
\vspace{-0.2in}
\begin{align*}
& \textbf{\DirMCRel}\\
\min \quad & \sum_{e \in E} w_e x_e & \\
   & \sum_{e \in p} x_e & \ge 1 & \qquad p \in \cP_{ij}, i \neq j \\
&  x_e & \ge 0 & \qquad e \in E
\end{align*}
\end{boxedminipage}
  \caption{LP Relaxation for \DirMC}
  \label{fig:dirmc-lp}
\end{figure}

The main result of the paper is the following theorem.

\begin{theorem}\label{thm:directed_cut_approximation}
  There is a randomized algorithm that given a
  feasible solution $\bx$ to {\sc \DirMCRel} returns a feasible
  integral solution of expected cost at most $2 \sum_e w_e x_e$, and
  runs in $O(m + n \log n)$ time. The algorithm can be 
  derandomized to yield a deterministic $2$-approximation algorithm 
  that runs in $O(m \log n)$ time. Here, $m = |E(G)|, n = |V(G)|$.
\end{theorem}

We now describe the simple randomized ball-cutting algorithm that
achieves the properties claimed by the theorem. Let $\bx$ be a
feasible solution to $\DirMCRel$. For any two nodes $u,v \in V$ we
define $d_x(u,v)$ be the shortest path length from $u$ to $v$ using
edge lengths given by $\bx$. For notational simplicity we omit the
subscript $x$ since there is little chance of confusion. The algorithm
adds new nodes $t_1,t_2,\ldots,t_k$ and adds the edge set $\{(t_i,s_j)
\mid i \neq j\}$ and sets the $x$ value of each of these new edges to
$0$. Note that, this is in effect a reduction of the \DirMC for the
given instance to a \DirMulticut instance which requires us to
separate the pairs $(t_i,s_i)$, $1 \le i \le k$. The solution $\bx$
augmented with the extra nodes and edges leads to a feasible
fractional solution for this \DirMulticut instance. Our algorithm,
formally described below, is very simple. We pick a random $\theta \in
(0,1)$ and take the union of the cuts defined by balls of radius
$\theta$ around each $t_i$. More formally let $B(v,r)$ be the set of
all nodes at distance at most $r$ from $v$.  Then the algorithm simply
outputs $\bigcup_{i=1}^k \delta^+(B(t_i,\theta))$ where $\delta^+(A)$
denote the set of outgoing edges from $A$.

\begin{algorithm}
	\caption{Rounding for \DirMC}
	\label{alg:directed_cut_rounding_scheme}
	\begin{algorithmic}[1]
		\STATE Given a feasible solution $\bx$
                to \DirMCRel
		\STATE Add new vertices $t_1,\dots,t_k$, edges $(t_i,s_j)$ 
              for all $i \neq j$ and set $x(t_i,s_j) = 0$ 
		\STATE Pick $\theta \in (0,1)$ uniformly at random
                \STATE $C = \cup_{i=1}^k \delta^+(B(t_i,\theta))$
		\STATE Return $C$
	\end{algorithmic}
\end{algorithm}

Note that $C$ is a random set of edges that depends on the choice of
$\theta$. We denote by $C(\theta)$ the set of edges output by the
algorithm for a given $\theta$.

\begin{lemma}
  \label{lem:dir-feasibility}
  If $\bx$ is a feasible fractional solution to \DirMCRel, $C(\theta)$
  is a feasible multiway cut for $\{s_1,\ldots,s_k\}$ for any $\theta
  \in (0,1)$. Thus, Algorithm \ref{alg:directed_cut_rounding_scheme}
  always returns a feasible integral solution given a feasible $\bx$.
\end{lemma}
\begin{proof}
  Fix any $i \in \{1,\ldots, k\}$ and $\theta \in (0,1)$. Since
  $d(t_i,s_j) = 0$ for all $j \neq i$, we have that $s_j \in
  B(t_i,\theta)$ for all $j \neq i$. Moreover, by feasibility of
  $\bx$, we have $d(t_i,s_i) \ge 1$ for otherwise there will be a path
  of length less than $1$ from some $s_j$ to $s_i$ where $j \neq i$.
  Therefore $s_i \not \in B(t_i,\theta)$ because $\theta < 1$.
  Therefore, $G - \delta^+(B(t_i,\theta))$ has no path from $s_j$ to
  $s_i$ for any $j \neq i$. Since $C(\theta) = \bigcup_i
  \delta^+(B(t_i,\theta))$, it follows that there is no path in $G -
  C(\theta)$ from $s_j$ to $s_i$ for any $j \neq i$.
\end{proof}

We now bound the probability that any fixed edge $e$ is cut by the
algorithm, that is, $\Pr[e \in C]$.  Note that $e$ may be {\em
  simultaneously} cut by several $t_i$ for the same value of $\theta$
but we are only interested in the probability that it is included in
$C$.

\begin{lemma}
  \label{lem:prob-e-cut}
  For any edge $e \in E$, $\Pr[e \in C] \le 2 x_e$.
\end{lemma}
\begin{proof}
  Let $e=(u,v)$. Rename the terminals such that $d(s_1,u) \leq
  d(s_2,u) \leq \dots \leq d(s_k,u)$. This implies that 
  $$d(t_1,u) = d(s_2,u)$$
  and 
  $$d(t_2,u) = d(t_3,u) = \ldots = d(t_k,u) = d(s_1,u).$$
  Edge $e \in \delta^+(B(t_i,\theta))$ if and only if $\theta \in [d(t_i,u),
  d(t_i,v))$; we have that $d(t_i,v) \le d(t_i,u) + x_e$. 
  Defining the interval $I_i$ as $[d(t_i,u),d(t_i,u)+x_e)$, we see
  that $e \in \delta^+(B(t_i,\theta))$ only if $\theta \in I_i$. 
  However, from the property that $d(t_2,u) = d(t_3,u) \ldots = d(t_k,u)$,
  $I_2 = I_3 = \ldots = I_k$. Thus, 
  $e \in C$ only if $\theta \in I_1$ or $\theta \in I_2$ and
  since $|I_1|$ and $|I_2|$ are both at most $x_e$ long and $\theta$
  is chosen uniformly at random from $(0,1)$,
  $$\Pr[e \in C] \le \Pr[\theta \in I_1] + \Pr[\theta \in I_2] \le 2 x_e.$$
\end{proof}

\begin{corollary}
  $\Ex[C]$, the expected cost of $C$, is at most $2 \sum_e w_e x_e$.
\end{corollary}

\paragraph{Running time analysis and derandomization:}
A natural implementation of Algorithm
\ref{alg:directed_cut_rounding_scheme} would first choose $\theta$ and
then compute $\delta^+(B(t_i,\theta))$ for each $i$.  This can be easily
accomplished via $k$ executions of Dijkstra's single-source shortest
path algorithm, one for each $t_i$, leading to a running time of
$O(k(m + n \log n))$ where $m = |E|$ and $n=|V|$. However, by taking
advantage of our analysis in Lemma~\ref{lem:prob-e-cut}, we can obtain
a run time that is equivalent to a single execution of Dijkstra's
algorithm.

Consider a slight variation of Algorithm
\ref{alg:directed_cut_rounding_scheme}. For each edge $e=(u,v)$,
define two intervals $I_1(e) = [d(s_1,u), d(s_1,u) + x_e)$ and $I_2(e)
= [d(s_1,u), d(s_1,u) + x_e)$, where $s_1,s_2$ are the two terminals
from which $u$ is the closest in terms of distance. We pick $\theta
\in (0,1)$ uniformly at random and include $e$ in $C$ iff $\theta \in
I_1(e)$ or $\theta \in I_2(e)$. The analysis in Lemmas
\ref{lem:dir-feasibility} and \ref{lem:prob-e-cut} shows that even
this modified algorithm outputs a feasible cut whose expected cost is at
most $2\sum_e w_e x_e$. Note that the edges cut by this modified
algorithm may be a strict superset of the edges cut by Algorithm
\ref{alg:directed_cut_rounding_scheme}. The advantage of the modified
algorithm is that we only need to calculate $I_1(e)$ and $I_2(e)$ for
each edge $e \in E$. To do this, for each node $u$, we need to find
the two terminals from which $u$ is the closest and their
corresponding distances. More formally, consider the following
$h$-nearest-terminal problem.

\begin{prob}
  Given a directed graph $G=(V,E)$ with non-negative edge-lengths, a set $S
  \subseteq V(G)$ of $k$ terminals, and an integer $h \le k$, for each
  vertex $v$, find the $h$ terminals from which $v$ is the closest
  among the terminals and their corresponding distances. In other
  words for each $v$ find the $h$ smallest values in 
  $d(s_1,v),d(s_2,v),\ldots,d(s_k,v)$ where $S = \{s_1,\ldots,s_k\}$.
\end{prob}

The above problem can be solved via a randomized algorithm using
hashing that runs in expected time $O(h (m + n \log n))$, which
corresponds to $h$ executions of Dijkstra's algorithm.  It can also be
solved in $O(h m\log h + h n \log n)$ time via a deterministic
algorithm. See \cite{HarPeled15} who refers to this as the
$h$-nearest-neighbors problem.

Using the algorithm for the $h$-nearest-terminal problem with $h=2$,
we can calculate $I_1(e)$ and $I_2(e)$ for each $e \in E$ in $O(m + n
\log n)$ time\footnote{One can easily derive the $h=2$ case from first
  principles also.}.  We then chose $\theta$ uniformly at random from
$(0,1)$ and cut $e$ if $\theta$ lies in one of the range $I_1(e)$ or
$I_2(e)$. This gives us a $2$-approximate randomized algorithm with
running time $O(m + n \log n)$.

We can derandomize the algorithm by computing the cheapest cut among
all $\theta \in (0,1)$ as follows. Once $I_1(e)$ and $I_2(e)$ are
computed for each $e$ we sort the $4m$ end points of these $2m$
intervals; let them be $\theta_1 \le \theta_2 \le \ldots \le
\theta_{4m}$.  We observe that it suffices to evaluate the cut value
at each of these values of $\theta$.  A simple scan of these $4m$
points while updating the cut-value at each end point can be
accomplished in $O(m)$ time. Sorting the end points takes $O(m \log
n)$ time.  This leads to a deterministic $2$-approximation algorithm
with running time $O(m \log n)$.

\subsection{\DirMC with $k=2$}
\label{sec:2-terminal-dir-mc}
In this section we address \DirMC with $k=2$ which we refer to as
\2DirMC. We believe this is an interesting problem on its own as it is
related closely to the classical $s$-$t$ cut problem.  As we remarked
earlier, \2DirMC is NP-Hard and APX-Hard to approximate. This was
shown in \cite{GargVY94,GargVY04} via a simple approximation preserving
reduction from \MC with $k=3$.  Another consequence of the reduction
is that the integrality gap of \DirMCRel for \2DirMC is at least
$4/3$. On the other hand no ratio better than $2$ is known for
\2DirMC. This naturally raises the following question.

\begin{question}
  What is the integrality gap of \DirMCRel for \2DirMC? What is
  the approximability of \2DirMC?
\end{question}

We obtain two theorems. The first one shows that the integrality gap 
for \2DirMC is $2$.

\begin{theorem}\label{thm:2-terminal-directed-integrality-gap}
  Integrality gap of \DirMCRel for \2DirMC is $2$ even in planar directed
  graphs.
\end{theorem}

The second theorem slightly extends a result in \cite{GargVY04}.

\begin{theorem}
\label{thm:2-terminal-hardness}
  There is an approximation preserving reduction from $4$-terminal
  \NodeMC to \2DirMC.
\end{theorem}

We raise the following question.

\begin{question}
  Can we prove a factor $2$ hardness of approximation for \DirMC under
  the assumption that $P \neq NP$? Does a factor of $2$ hardness hold
  for \2DirMC even under the Unique Games conjecture?
\end{question}

\paragraph{Integrality gap construction:} 
Proof of Theorem~\ref{thm:2-terminal-directed-integrality-gap}
is based on recursively defined sequence of graphs
$G_0,G_1,\ldots,G_h$ with increasing integrality gap; we will use
$\alpha_i$ to denote the integrality gap (we also refer to this
as the flow-cut gap) in $G_i$.  The two terminals
will be denoted by $s,t$. The symmetry in the construction will ensure
that in $G_i$ the $s$-$t$ cut value will be equal to the $t$-$s$ cut
value; we refer to these common values as the one-way cut value and
the optimum value of a cut that separates $s$ from $t$ and $t$ from
$s$ as the two-way cut value. The graph $G_0$ is shown in
Fig~\ref{fig:gap} and it is easy to see that $\alpha_0 = 1$.

\begin{figure}[htb]
\centering
\includegraphics[scale=0.8]{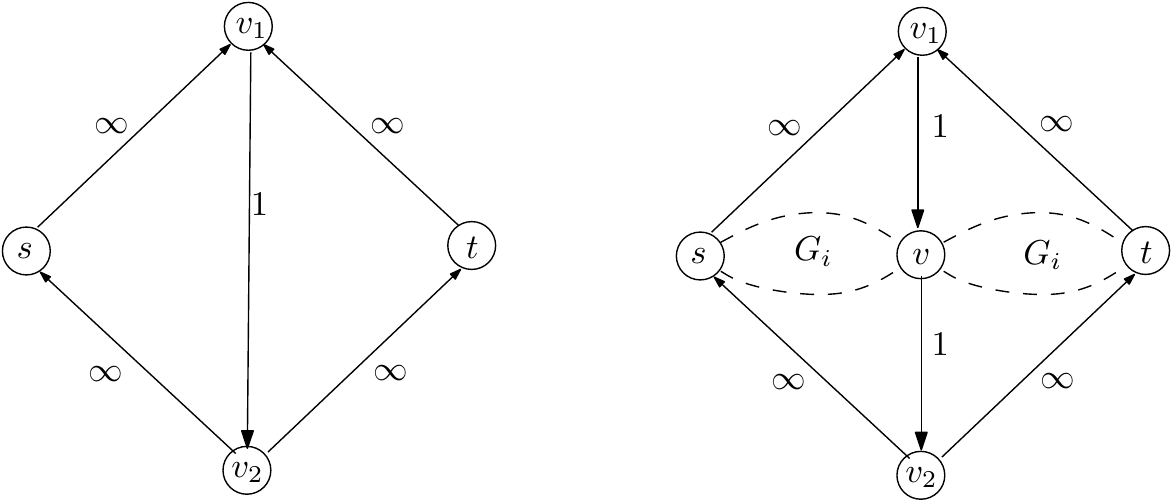}
\caption{$G_{0}$ on the left and constructing $G_{i+1}$ from $G_i$ shown
on the right.}
\label{fig:gap}
\end{figure}

The iterative construction of $G_{i+1}$ from $G_i$ is shown at a
high-level in figure \ref{fig:gap}. A formal description is as
follows. To obtain $G_{i+1}$ with terminals $s,t$ we start with two
copies of $G_i$ with terminals $s_1,t_1$ and $s_2,t_2$ (denoted 
by $H,H'$) and two new vertices $v_1,v_2$. We set $s=s_1$, $t=t_2$ and
identify $t_1$ and $s_1$ as the center vertex $v$ shown in the
figure. We add edges $(v_1,v)$ and $(v,v_2)$ with weight $1$ and four
other edges $\{(s,v_1),(t,v_1),(v_2,s),(v_2,t)\}$ each with weight
infinity. Finally we scale the weights of the edges of $H$ and $H'$
such that the two-way cut value in each of them is
$\frac{\alpha_i}{2-\alpha_i}$. It is easy to observe inductively that
the each graph in the sequence is planar and moreover the graph can be
embedded such that $s$ and $t$ are on the outer face.  The analysis of
the integrality gap of this construction can be found in the appendix.

Subsequent to our construction, Julia Chuzhoy obtained
an alternative non-recursive construction with an integrality gap of
$2$ for \2DirMC. 

\paragraph{Reduction from $4$-terminal \NodeMC to \2DirMC:} 
Given a \NodeMC instance with graph $G$ and set of terminals
$\{s_1,s_2,s_3,s_4\}$, Figure~\ref{fig:node_to_directed_reduction}
shows the ingredients of a reduction to \DirMC instance with graph
$G'$ and terminals $s,t$. This is a slight modification of the
reduction from three-terminal \MC to \2DirMC given in \cite{GargVY04}.
It is convenient to consider the node-weighted
version of \DirMC which is equivalent to the edge-weighted version.
Formally $G'$ is obtained from $G$ by the addition of two new nodes
$s,t$ which are connected to the terminals via directed edges of
infinite weight as shown in the figure. Each edge $uv \in E(G)$ is
replaced by two directed edges $(u,v)$ and $(v,u)$ and the weights of
the nodes of $G$ remain the same. We will assume without loss of generality
that the terminals $s_1,s_2,s_3,s_4$ have infinite weight.
A relatively simple case analysis shows that 
$C \subset V(G)$ is a feasible node-multiway cut for the terminals
$\{s_1,\ldots,s_4\}$ in $G$ iff $C$ is a feasible node-multiway cut in 
$G'$ for $\{s,t\}$. This type of reduction does not seem to generalize
beyond four terminals.

\begin{figure}[htb]
\centering
\includegraphics[scale=1]{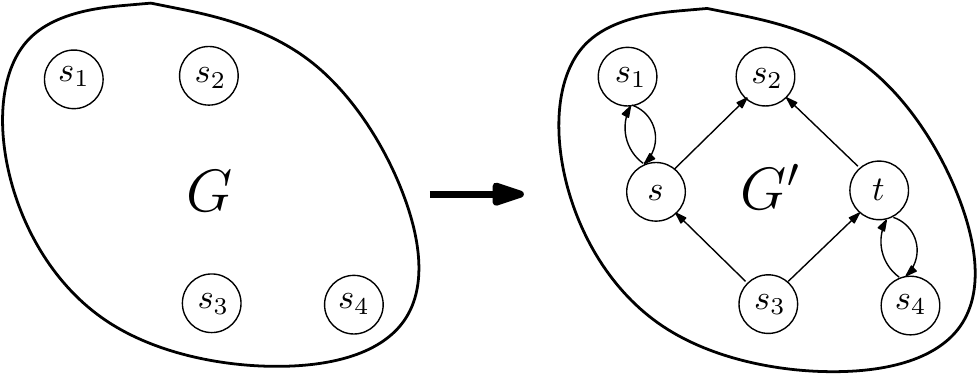}
\caption{Reduction from 4-terminal \NodeMC to \2DirMC. Non-terminal
vertices are not shown.}
\label{fig:node_to_directed_reduction}
\end{figure}

Garg \etal \cite{GargVY04} showed that \DirMCRel does not necessarily
have half-integral opitmum solutions. In
Section~\ref{sec:fractionality} we extend their example to show that
for every non-negative integer $\ell$ there exist instances for which 
there is no optimum solution to \DirMCRel that is $1/\ell$ integral.

\section{LP Relaxation and rounding for \NodeMC}
\label{sec:node-mc}

The LP relaxation for the \NodeMC
is similar to the one for \MC. We
have a variable $x_v \in \{0,1\}$ for each $v \in V$ which indicates
whether to remove $v$ or not. We can assume without loss of generality
that we cannot remove the terminals $s_1,s_2,\ldots,s_k$ and moreover
that they form an independent set. This can be accomplished by adding
to each original terminal $s_i$ a new dummy terminal $s'_i$ and
adding the edge $s_is'_i$. Let $\cP_{ij}$ be the set of all 
paths between $s_i$ and $s_j$ in $G$. Note that in the undirected graph
case we do not need to distinguish $P_{ij}$ from $P_{ji}$. Let $S = \{s_1,s_2,\dots,s_k\}$ be the set of terminals.
\begin{figure}[htb]
  \centering
\begin{boxedminipage}{0.5\linewidth}
\vspace{-0.2in}
\begin{align*}
& \textbf{\NodeMCRel}\\
\min \quad & \sum_{v \in V\setminus S} w_v x_v & \\
   & \sum_{v \in p} x_v & \ge 1 & \qquad p \in \cP_{ij}, i < j \\
   & x_v & = 0 & \qquad v \in S \\
&  x_v & \ge 0 & \qquad v \in V
\end{align*}
\end{boxedminipage}
  \caption{LP Relaxation for \NodeMC}
  \label{fig:nodemc-lp}
\end{figure}

\begin{theorem}\label{thm:node_cut_approximation}
  There is a polynomial-time randomized algorithm that given a
  feasible solution $\bx$ to {\sc \NodeMCRel} returns a feasible
  integral solution of expected cost at most $2 (1-1/k)\sum_v w_v
  x_v$, and runs in $O(m + n \log n)$ time.  The algorithm can be
  derandomized to yield a deterministic $2$-approximation algorithm
  that runs in $O(kn+ m + n\log n)$ time.
\end{theorem}

Let $\bx$ be a feasible fractional solution to \NodeMCRel.  For nodes
$u$ and $v$ we define $d_x(u,v)$ to be the length of the shortest path
between $u$ and $v$ according to the node weights given by $\bx$; we
count the weights of the end points $u$ and $v$ in $d_x(u,v)$. We omit
the subscript $x$ in subsequent discussion.  For a given radius $r$
and node $u$ let $B(u,r)$ be the set of all nodes $v$ such that
$d(u,v) \le r$; $B(u,r)$ is the ball of radius $r$ around $u$. 
We define the ``boundary'' of radius $r$ from $u$, denoted by
$B^+(u,r)$ to be the set of all nodes that are not in $B(u,r)$ but
have an edge to some node in $B(u,r)$.

\begin{prop}
  \label{prop:node-boundary}
  A node $v \in B^+(u,r)$ iff $r < d(u,v) \le r + x_v$. Further, 
  if $v \in B^+(u,r)$ for $r < 1$ then $x_v \neq 0$.
\end{prop}

Our rounding algorithm first picks an index $\ell$ uniformly at random
from $\{1,2,\ldots,k\}$. It then picks a $\theta$ uniformly at random
from $(0,1/2)$. For each $i \neq \ell$ it includes in the final cut $C$
all nodes $v$ that are in the ``boundary'' of the ball of radius
$\theta$ around $s_i$.  The formal description is given in
Algorithm~\ref{alg:node_cut_rounding_scheme}.

\begin{algorithm}[htb]
	\caption{Rounding for \NodeMC}
	\label{alg:node_cut_rounding_scheme}
	\begin{algorithmic}[1]
		\STATE Given feasible fractional solution $\bx$ to \NodeMCRel
		\STATE Chose $\ell \in \{1,2,\dots,k\}$ uniformly at random
		\STATE Pick $\theta \in (0,1/2)$ uniformly at random
                \STATE $C = \cup_{i \neq \ell} B^+(s_i,\theta)$
		\STATE Return $C$
	\end{algorithmic}
\end{algorithm}

Let $C(\ell,\theta)$ be the output of the algorithm for fixed $\ell$
and $\theta$.  We first argue that the algorithm always returns a
feasible multiway cut.
\begin{lemma}
For all $\ell,\theta$, $C(\ell,\theta)$ 
is a feasible multiway cut for the given instance. That is, $G - C(\ell,\theta)$
has no path from $s_i$ to $s_j$ for $i \neq j$.
\end{lemma}
\begin{proof}
  Consider any pair $i,j \in \{1,2,\ldots,k\}$ where $i \neq j$.
  Assume $i \neq \ell$, the case when $j \neq \ell$ is similar.
  The ball $B(s_i,\theta)$ does not contain $s_j$ since 
  $\theta < 1/2$ and $d(s_i,s_j) \ge 1$ by feasibility of $\bx$.
  $C(\ell,\theta)$ contains all nodes from $B^+(s_i,\theta)$, thus, in  $G-C(\ell,\theta)$ there
  cannot be a path from $s_i$ to any node in $V \setminus B(s_i,\theta)$,
  and hence to $s_j$.
\end{proof}

We say that $v$ is cut by the algorithm if $v \in C$.  The key to the
performance guarantee of the algorithm is the following lemma.
\begin{lemma}
  \label{lem:node-cut-prob}
  $\Pr[v \in C] \leq 2(1-1/k) x_v$.
\end{lemma}
\begin{proof}
  Fix a node $v$ and rename the terminals such that $d(s_1,v) \leq
  d(s_2,v) \leq \dots \leq d(s_k,v)$. Define the interval $I_i$ as
  $[d(s_i,v) - x_v, \min(d(s_i,v),1/2))$. From the algorithm
  description and Proposition~\ref{prop:node-boundary}, we can see
  that $v \in C$ iff $\exists i$ such that $\ell \neq i$ and $\theta
  \in I_i$.

  Note that $I_i$ is an empty interval if $x_v = 0$ or $d(s_1,v) - x_v
  \geq 1/2$. Hence we can assume that $x_v > 0$ and $d(s_1,v) - x_v <
  1/2$, otherwise $I_i$ is empty for all $i$ and $\Pr[v \in C] =
  0$. We now consider two cases depending on
  whether $d(s_2,v) - x_v$ is greater than $1/2$ or not.

  First, consider the case when $d(s_2,v) - x_v \geq 1/2$. Interval
  $I_2$ is empty. Since $d(s_2,v) \leq d(s_3,v) \leq \dots \leq
  d(s_k,v)$, intervals $I_3,I_4,\dots,I_k$ are also empty. Hence, $v\in
  C$ iff $\ell \neq 1$ and $\theta \in I_1$. Interval $I_1$ has length at
  most $x_v$ and $\theta$ is chosen uniformly at random from
  $(0,1/2)$. Therefore,

\begin{equation*}
  \Pr[ v \in C] = \Pr[\ell \neq 1] \Pr[\theta \in I_1] \leq \big(1-\frac{1}{k}\big) \cdot 2x_v.
\end{equation*}
In the preceding equation we used independence in the choice of $\ell$ and 
$\theta$.

Next, consider the case when $d(s_2,v) -x_v < 1/2$. From the
feasibility of $\bx$, we have that $d(s_1,v) - x_v + d(s_2,v) \geq 1$
(recall that $d(s_1,v)$ and $d(s_2,v)$ include the length of
$x_v$). This implies that $d(s_1,v) \geq 1/2$. Since, $d(s_i,v) \geq
d(s_1,v)$ for all $i$, we have $d(s_i,v) \geq 1/2$ which implies that
for all $i$, $I_i =[d(s_i,v) - x_v, 1/2)$. Easy to see that $I_1
\supseteq I_2\dots \supseteq I_k$. Therefore, $v \in C$ iff $\ell = 1$
and $\theta \in I_2$ or $\ell \neq 1$ and $\theta \in I_1$. Length of
interval $I_1$ and $I_2$ are $1/2-d(s_1,v) + x_v$ and $1/2 - d(s_2,v)
+ x_v$ respectively. Hence,
\begin{eqnarray*}
\Pr[v \in C ] &=& \Pr[\ell = 1] \Pr[ \theta \in I_2] + \Pr[\ell \neq 1] \Pr[\theta \in I_1]\\
		& = & 1/k\cdot 2(1/2-d(s_2,v)+x_v) + (1-1/k)\cdot 2(1/2-d(s_1,v) + x_v) \\
		& \leq & 2(1-1/k)(1-d(s_1,v) -d(s_2,v)+ 2 x_v) \\
		& \leq & 2(1-1/k) x_v
\end{eqnarray*}
In the penultimate inequality above, we use the fact that $1-1/k \geq
1/k$ if $k \geq 2$. The final inequality follows from already stated
observation, $d(s_1,v) + d(s_2,v) - x_v \geq 1$ due to feasibility of
$\bx$.
\end{proof}

\begin{corollary}
  $\Ex[w(C)] = \sum_{v \in V} w_v \Pr[v \in C] \le 2(1-1/k) \sum_v w_v x_v$.
  Thus, the expected cost of the cut output by the algorithm is at most
  $2(1-1/k)$ times the cost of the fractional solution $\bx$.
\end{corollary}

\paragraph{Running time:} 
Algorithm \ref{alg:node_cut_rounding_scheme} can be implemented in
$O(m + n \log n)$ time, in a fashion very similar to the
implementation of the modified version of Algorithm
\ref{alg:directed_cut_rounding_scheme}. First, we pick $\ell$
uniformly at random from $\{1,\dots,k\}$ and $\theta$ uniformly at
random from $(0,1/2)$. Then, for each vertex $v$ we find the closest
terminal $s$ in the set $S \setminus \{s_\ell\}$ and cut vertex $v$ if
$d(s,v) - x_v \leq \theta < d(s,v)$. Finding nearest terminal for each
vertex can be done in $O(m + n \log n)$ time. Hence, we get a
randomized $2(1-1/k)$-approximation rounding scheme in time $O(m + n \log n)$.

To derandomize, we consider for each $v$ intervals $I_1(v)$ and
$I_2(v)$ as in the proof of Lemma~\ref{lem:node-cut-prob}. Using the
$h$-nearest terminal algorithm for $h = 2$ with $S$ as the set of
terminals, in $O(m + n \log n)$ time, we can compute $I_1(v)$ and
$I_2(v)$ for all $v$. We sort the $4n$ end points of these $2n$
intervals and let them be $\theta_1,\theta_2,\ldots,\theta_{4n}$.  It
suffices to find the cost of the cut for each $\theta$ from this $4n$
values and for each $\ell \in \{1,2,\ldots,k\}$.  We process these
sorted values in order and for each $\theta$, we calculate
$w(C(\ell,\theta))$ for all $\ell$.  The proof of
Lemma~\ref{lem:node-cut-prob} shows that this can be done by using
only $I_1(v)$ and $I_2(v)$ for all $v$.  As we process the end points
in the sorted order the time to update the cut for each $\ell$ per end
point is $O(1)$.  Thus, in $O(nk + m + n \log n)$ time we can obtain a
deterministic algorithm that gives a $2(1-1/k)$-approximation.

\paragraph{Acknowledgments:} CC thanks Sudeep Kamath, Sreeram Kannan and
Pramod Viswanath for extensive discussions on the problems considered
in \cite{KKCV15} which inspired us to revisit the rounding schemes
for multiway cut problems. CC also thanks Anupam Gupta
for discussion on some of the problems considered in \cite{KKCV15} during
an Oberwolfach workshop. 

\bibliographystyle{plain}
\bibliography{soda}

\appendix
\section{Proof of Theorem~\ref{thm:2-terminal-directed-integrality-gap}}
Here we prove the correctness of the integrality gap construction
described in Section~\ref{sec:2-terminal-dir-mc}.

The following proposition is easy to establish based on the symmetry in
the construction of the graphs.
\begin{prop}
  The $s$-$t$ cut value and the $t$-$s$ cut value in $G_{i+1}$ are the same.
\end{prop}

Now, we calculate $\alpha_{i+1}$ in terms of $\alpha_i$. We refer to the
copy of $G_i$ containing $s$ and $v$ with scaled capacities as $H$, and the one 
containing $v$ and $t$ as $H'$.

\begin{lemma}
  For $i \ge 0$, $\alpha_{i+1} = \frac{4 - \alpha_i}{3 - \alpha_i}$. 
  For $i \ge 0$, the ratio of of
  the one-way cut value to the two-way cut value in $G_{i}$ is
  $\frac{1}{\alpha_{i}}$.
\end{lemma}
\begin{proof}
  Proof by induction on $i$. For the base case we see that $\alpha_0 =
  1$ and in $G_0$ the one-way cut value and two-way cut value are both
  $1$ and hence the ratio is equal to $1 = \frac{1}{\alpha_0}$.
  
  We now prove the induction step. For this purpose we estimate the
  one-way cut value and the two-way cut value in $G_{i+1}$.

  \noindent
  \textbf{Minimum two-way cut:} Any finite value cut that separates
  $s$ from $t$ has to cut at least one of the two edges $(v_1,v), (v,v_2)$.
  We consider two cases.
  
  \noindent\textbf{Case 1:} Both $(v_1,v), (v,v_2)$ are cut. 
  To separate $s$ and $t$ it is best to pick a two-way cut between $s$
  and $v$ in $H$ (or symmetrically between $v$ and $t$ in $H'$).
  Thus the total cost is $2 + \frac{\alpha_i}{2-\alpha_i} = \frac{4 -
    \alpha_i}{2-\alpha_i}$.
  
  \noindent\textbf{Case 2:} Only one of the edges  $(v_1,v), (v,v_2)$ is cut.
  Without loss of generality this edge is $(v,v_2)$. Since $(v_1,v)$
  is not cut $s$ and $t$ can reach $v$ via $v_1$. Thus any two-way cut
  in $G$ needs to use a one-way cut in $H$ to separate $v$ from $s$
  and a one-way cut in $H'$ to separate $v$ from $t$. The cost of each
  of these one-way cuts is, by induction, $\frac{1}{\alpha_i} \cdot
  \frac{\alpha_i}{2-\alpha_i} = \frac{1}{2-\alpha_i}$. Thus the total
  cost is $1 + \frac{2}{2-\alpha_i} = \frac{4 -
    \alpha_i}{2-\alpha_i}$.

  In both cases the cost is the same and hence the optimal two-way cut in
  $G_{i+1}$ is $\frac{4 - \alpha_i}{2-\alpha_i}$.

  \noindent\textbf{Minimum one-way cut:} We now calculate one-way cut
  from $s$ to $t$. At least one of the edges $(v_1,v), (v,v_2)$ has to
  be cut. Also, either there is no path from $s$ to $v$ or no path from $v$ to
  $t$. Thus, the cost of the one-way cut from $s$ to $t$ is at least $1
  + \frac{1}{2-\alpha_i} = \frac{3-\alpha_i}{2-\alpha_i}$. Moreover it
  is easy to see that this is achievable by removing $(v_1,v)$ and
  one-way cut from $s$ to $v$ in $H$.

  \noindent\textbf{Optimum fractional solution value:} We now 
  calculate the optimum for \DirMCRel on $G_{i+1}$. We consider the
  following feasible solution $x$. Assign $0$ to the infinite weight
  edges and $1/2$ to each of edges $(v_1,v)$ and $(v,v_2)$. For the
  edges in the graphs $H$ and $H'$ we take an optimum solution $y$ to
  \DirMCRel on $G_i$ and scale it down by $1/2$ and assign these
  values to the edges of $H$ and $H'$. Feasibility of $y$ for $G_i$
  implies that distance from $s$ to $v$ and $v$ to $s$ in $H$
  according to $x$ is $1/2$ (since we scaled down by $1/2$). It is
  easy to verify that distance of $s$ to $t$ and from $t$ to $s$ is
  $1$ in the fractional solution $x$ in $G_{i+1}$. Now we analyze the
  cost of this solution $\sum_{e \in E(G_{i+1})} w_e x_e$.  We have a
  total contribution of $1$ from the two edges $(v_1,v)$ and
  $(v,v_2)$.  We claim that $\sum_{e \in E(H)} w_e x_e = \frac{1}{2}
  \cdot \frac{1}{\alpha_i} \cdot \frac{\alpha_i}{2-\alpha_i}$ since
  the cost of the two-way cut in $H$ is chosen to be
  $\frac{\alpha_i}{2-\alpha_i}$, the integrality gap is
  $\alpha_i$ and we scaled down $y$ by $1/2$ to obtain $x$
  in $H$. Same holds for $H'$.  Thus the total fractional cost of this
  solution is $1 + \frac{1}{2-\alpha_i} = \frac{3- \alpha_i}{2-
    \alpha_i}$. We can see that this is an optimum solution by
  exhibiting a multicommodity flow of the same value for the pairs
  $(s,t)$ and $(t,s)$ in $G_{i+1}$. Route one unit of flow from $s$ to
  $t$ along the path $s \rightarrow v_1 \rightarrow v \rightarrow v_2
  \rightarrow t$. In $H$ there exists a feasible flow of total value
  $\frac{1}{\alpha_i} \cdot \frac{\alpha_i}{2-\alpha_i} =
  \frac{1}{2-\alpha_i}$. Let $f(s,v)$ and $f(v,s)$ be the amount of
  flow from $s$ to $v$ and $v$ to $s$ respectively. By duplicating this
  flow in $H'$ we see that a flow of value $\frac{1}{2-\alpha_i}$
  exists between $s$ and $t$ in $G_{i+1}$ via $H$ and $H'$. Thus there
  is a total flow of value at least $1 + \frac{1}{2-\alpha_i}$ in $G_{i+1}$
  and this is optimal.

  \medskip We can now put together the preceding bounds to prove the
  lemma.  The flow-cut gap in $G_{i+1}$ is seen to be the ration of
  the two-way cut value $\frac{4 - \alpha_i}{2-\alpha_i}$ and the
  maximum flow value $\frac{3- \alpha_i}{2- \alpha_i}$. Hence
  $\alpha_{i+1} = \frac{4 - \alpha_i}{3 - \alpha_i}$ as desired.  The
  ratio of one-way cut value $\frac{3 - \alpha_i}{2-\alpha_i}$ and the
  two-way cut value $\frac{4 - \alpha_i}{2-\alpha_i}$ in $G_{i+1}$ is
  $\frac{3 - \alpha_i}{4 - \alpha_i}$ which is equal to
  $\frac{1}{\alpha_{i+1}}$. This completes the inductive proof.
\end{proof}

\medskip
We have a sequence of numbers $\alpha_i$ where $\alpha_0 = 1$ and
$\alpha_{i+1} = \frac{4-\alpha_i}{3 - \alpha_i}$. It is easy to argue
that this sequence converges to $2$. This proves that the integrality
gap of \DirMCRel is in the limit equal to $2$.

\section{Fractionality of the LP solutions}
\label{sec:fractionality}
It was shown in \cite{GargVY04} that there is a half-integral optimum
solution for the natural LP relaxation for node-weighted multiway cut
(\NodeMC) which was then exploited to obtain a
$2(1-1/k)$-approximation.  \cite{GargVY04} also showed that the
half-integral property does not hold for \2DirMC. Here we generalize
their example to observe that for any positive integer $\ell$ there are
examples where there may not exist an optimum solution to \DirMCRel on
instances with two terminals that is $1/\ell$ integral. More generally,
there does not exists an edge with length more than $1/\ell$.

Consider the generalization of the example in \cite{GargVY04} as shown
in Fig~\ref{fig:2-terminal-non-half-integral-example}.  Each flow path
from $s$ to $t$ or $t$ to $s$ has to use at least $h$ edges of the
type $(u_i,u_{i+1})$ or $(v_j,v_{j+1})$. Since, there are only
$2(h-1)$ such edges, flow is upper bounded by $2(h-1)/h$. To see that
this flow is also achievable, consider the following sets of
paths. For $1\leq i \leq h-1$, path $P_i = s, u_1,\ldots,u_{i+1},
v_i,\ldots,v_h,t$ and path $P_i' =
t,v_1,\ldots,v_{i+1},u_i,\ldots,u_h,s$. Send $1/h$ unit of flow along
each of these paths. Each of the edge $(u_j,u_{j+1})$ is part of $P_i$
for $i \geq j$ and part of $P_i'$ for $i \leq h-j$. Hence, capacity
used for edge $(u_i,u_{i+1})$ is $h\cdot 1/h=1$. Similarly for each
edge $(v_i,v_{i+1})$. Flow value is equal to $2(h-1)/h$. So, optimum
solution has value $2(h-1)/h$.

\begin{figure}[hbt]
\centering
\includegraphics{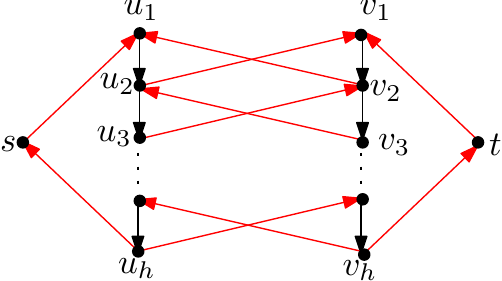}
\caption{Edges of the form $(u_i,u_{i+1})$ or $(v_j,v_{j+1})$ have capacity $1$
  and rest have infinite capacity. Optimal fractional cut/flow is $2(1-1/h)$.}
\label{fig:2-terminal-non-half-integral-example}
\end{figure}

By strong duality, optimal value of \DirMCRel is equal to maximum flow which is
equal to $2(h-1)/h$. Let $x$ be an optimal solution to the
\DirMCRel. By feasibility of the solution, each of the paths $P_i$ and
$P_i'$ has length at least $1$. Summing up the lengths of path $P_i$
and $P_i'$, we get $\left(\sum_{j=1}^{h-1} (x(u_j,u_{j+1}) +
  x(v_j,v_{j+1}))\right) + x(u_i,u_{i+1}) + x(v_i,v_{i+1}) \geq 2$. By
optimality of the solution first term is equal to
$2(h-1)/h$. Therefore, $x(u_i,u_{i+1}) + x(v_i,v_{i+1}) \geq
2/h$. Since, this inequality holds for all $1\leq i \leq h-1$, and
$\sum_{j = 1}^{h-1} (x(u_j,u_{j+1}) + x(v_j,v_{j+1})) = 2(h-1)/h$, we get that all
the inequalities are tight and $x(u_i,u_{i+1}) + x(v_i,v_{i+1}) =
2/h$. Since, all lengths are non-negative, $x(u_i,u_{i+1}),
x(v_i,v_{i+1}) \leq 2/h$. By taking $h > 2\ell$, we get an instance where
optimal solution has no edge having length at least $1/\ell$.

\end{document}